\documentclass[submission,copyright,creativecommons]{eptcs}
\providecommand{\event}{SynCoP 2014} % Name of the event you are submitting to
\usepackage{breakurl}             % Not needed if you use pdflatex only.
\usepackage{mathbbol}
\usepackage{newproof}
\usepackage{amsmath}
\usepackage{gastex}

\newcommand{\problemx}[3]{
\vspace{3mm}
\par\noindent{\bf#1}\par\nobreak\vskip.2\baselineskip
\begingroup\clubpenalty10000\widowpenalty10000
\setbox0\hbox{\bf INPUT:\ \ }\setbox1\hbox{\bf QUESTION:\ \ }
\dimen0=\wd0\ifnum\wd1>\dimen0\dimen0=\wd1\fi
\vskip-\parskip\noindent
\hbox to\dimen0{\box0\hfil}\hangindent\dimen0\hangafter1\ignorespaces#2\par
\vskip-\parskip\noindent
\hbox to\dimen0{\box1\hfil}\hangindent\dimen0\hangafter1\ignorespaces#3\par
\endgroup\vspace{3mm}
}

\newcommand{\N}{\mathbb{N}}
\newcommand{\RP}{\mathbb{R}_{\geq 0}}
\newcommand{\QP}{\mathbb{Q}_{\geq 0}}

\newcommand{\finally}{\mathsf{F}}

\newcommand{\locs}{\mathcal{L}}
\newcommand{\edges}{E}

\newcommand{\clocks}{\mathcal{X}}
\newcommand{\clock}{x}

\newcommand{\loc}{\mathit{l}}

\newcommand{\TW}{T\Sigma^+}

\usepackage{todonotes}
\newcommand{\true}{\mathtt{true}}
\newcommand{\false}{\mathtt{false}}

\newcommand{\U}{\mathsf{U}}
\newcommand{\X}{\mathsf{X}}

\newcommand{\sep}{\textrm{ }|\textrm{ }} % grammar separator

\renewcommand{\frac}{\mathsf{frac}}
\renewcommand{\int}{\mathsf{int}}

\newcommand*\ie{\textit{i.e.}}
	\newcommand*\eg{\textit{eg.}}

\newcommand{\glob}{\mathsf{G}}

\newcommand{\A}{\mathcal{A}}

\newcommand{\cm}{\mathcal{C}}

\newcommand{\pv}{\pi} % parameter valuation
\newcommand{\params}{\mathcal{P}}

\makeatletter
\newcommand*{\defeq}{\mathrel{\rlap{%
                     \raisebox{0.3ex}{$\m@th\cdot$}}%
                     \raisebox{-0.3ex}{$\m@th\cdot$}}%
                     =}
\makeatother

\makeatletter
\newcommand*{\ndefeq}{\mathrel{\rlap{%
                     \raisebox{0.3ex}{$\m@th\cdot$}}%
                     \raisebox{-0.3ex}{$\m@th\cdot$}%
                     \rlap{%
                     \raisebox{0.3ex}{$\m@th\cdot$}}
                     \raisebox{-0.3ex}{$\m@th\cdot$}}
                     =}
\makeatother

\newtheorem{theorem}{Theorem}[section]
\newtheorem{example}{Example}[section]
\newtheorem{lemma}[theorem]{Lemma}
\newtheorem{corollary}[theorem]{Corollary}

\title{MTL-Model Checking of One-Clock Parametric Timed Automata is Undecidable}
\author{Karin Quaas\thanks{The author is supported by Deutsche Forschungsgemeinschaft (DFG), project QU~316/1-1.}
\institute{Institut f\"ur Informatik\\ Universit\"at Leipzig \\D-04109 Leipzig, Germany}
}
\def\titlerunning{MTL-Model Checking of One-Clock Parametric Timed Automata is Undecidable}
\def\authorrunning{K. Quaas}
\begin{document}
\maketitle

\begin{abstract}
	Parametric timed automata extend timed automata (Alur and Dill, 1991) 
	in that %the constraints labelling the edges of the automaton allow 
	they 
	allow the specification of \emph{parametric} bounds on the clock values.
	Since their introduction in 1993 
	by Alur, Henzinger, and Vardi,
	it is known that the emptiness problem for parametric timed automata with one clock is decidable, whereas it is undecidable if the automaton  uses three or more parametric clocks. The problem is open for parametric timed automata with two parametric clocks. 
	Metric temporal logic, MTL for short, is a widely used specification language for real-time systems. MTL-model checking of timed automata is decidable, no matter how many clocks are used in the timed automaton.  
	In this paper, we prove that MTL-model checking for parametric timed automata is undecidable, even if the automaton  uses only one clock and one parameter and is deterministic. 
\end{abstract}

\section{Introduction}
An important field of algorithmic verification is the analysis of real-time systems, \ie, systems whose behaviour depend on time-critical aspects. 
Since the early nineties, 
numerous formalisms have been investigated to %satisfy the need to 
express and verify real-time properties. 
Two prominent examples of such formalisms are \emph{timed automata} and \emph{metric temporal logic}. 
Timed automata~\cite{AD94} extend classical finite automata with a finite set of real-valued \emph{clocks} whose values grow with the passage of time. The edges of a timed automaton are labelled with \emph{clock constraints} that compare the value of a clock with some constant. An edge can only be taken if the current values of the clocks satisfy the clock constraint labelling the edge. 
The central property of timed automata is the decidability of the emptiness problem~\cite{AD94}. 

Metric temporal logic (MTL, for short) extends classical linear temporal logic 
by constraining the temporal 
modalities with intervals of the non-negative reals. For example, the formula $\finally_{[0,2]}\varphi$ means that $\varphi$ will hold within two time units from now. 
Introduced by Koymans in 1990~\cite{K90}, 
the satisfiability problem and the model checking problem for timed automata were assumed to be undecidable for a long time. 
However, more than 20 years later it was proved by Ouaknine and Worrell~\cite{DBLP:journals/lmcs/OuaknineW07} that both problems are decidable if MTL is interpreted in the pointwise semantics over \emph{finite} timed words. 
The decidability of the MTL-model checking problem for timed automata is independent of the number of clocks that the timed automaton uses. 

A major drawback of timed automata and MTL is that they only allow the specification of \emph{concrete} constraints on timing properties, \ie, one has to provide the concrete values of all time-related constraints that occur in the real-time system.  
However, it is often more realistic to provide \emph{symbolic} (or, \emph{parametric}) constraints, in particular, if the real-time system under construction is not known in full details in the early stages of design.
With the purpose to overcome the incapability of timed automata to express  parametric time constraints, 
\emph{parametric timed automata} were introduced~\cite{DBLP:conf/stoc/AlurHV93}.
Parametric timed automata are timed automata defined over a finite set of parameters, which can be used in clock constraints labelling the edges of the automaton. 
For an example, consider the parametric timed automaton shown in Fig.Ê\ref{figure_ex} on page 4. 
The clock $y$ is concretely constrained by a constant like in ordinary timed automata. In contrast to this, the clock $x$ is parametrically constrained by the parameter $p$. The value of $p$ is determined by a parameter valuation, \ie, a function mapping each parameter to a value in the non-negative reals. 

A crucial verification problem for parametric timed automata is the emptiness problem: given a parametric timed automaton $\A$, does there exist some parameter valuation such that $\A$ has an accepting run? However,
it turns out that this problem is undecidable already if $\A$ uses three or more parametric clocks~\cite{DBLP:conf/stoc/AlurHV93}. On the positive side, the problem is decidable if in $\A$ at most one clock is compared to parameters. So far nothing is known about the decidability status for parametric timed automata with two parametric clocks; the problem is closely related to some hard and open problems of logic and automata theory~\cite{DBLP:conf/stoc/AlurHV93}. 

In this paper, 
we concern ourselves with the MTL-model checking problem for parametric timed automata: given a parametric timed automaton $\A$ and a specification in form of an MTL formula $\varphi$, does there exist some parameter valuation such that all finite runs of $\A$ satisfy $\varphi$?
For parametric timed automata with three clocks, the undecidability of this problem follows from the undecidability of the emptiness problem. 
Here, we prove that the problem is undecidable even if $\A$ uses only one clock and one parameter and is deterministic. 
This negative result is in contrast to the decidability of the emptiness problem for one-clock parametric timed automata, and the decidability of MTL-model checking of timed automata. 
The result can 
be regarded as further step towards the precise decidability border for the reachability problem for parametric timed automata with two parametric clocks, which is open for more than 20 years.

\paragraph{\em Related work}
% what about MTL pw und L+U?
The reader might wonder why we consider model checking for \emph{parametric} timed automata and \emph{standard} MTL, \ie, a non-parametric extension of MTL.
It is well known that if we extend classical LTL with formulae of the form $\varphi_1\U_{=p}\varphi_2$, meaning  that $\varphi_2$ has to hold in exactly $p$ steps from now on for some parameter $p$, then
the satisfiability problem (``Given a formula $\varphi$, is there some parameter valuation such that $\varphi$ is satisfiable?'') is undecidable: LTL with  parameterized \emph{equality modalities} of the form $\U_{=p}$ can be used to encode halting computations of two-counter machines~\cite{DBLP:journals/tocl/AlurETP01}. Undecidablity of the satisfiability problem implies undecidability of the model checking problem for all systems that are capable to recognize the universal language over a given alphabet (as it is the case for, \eg, timed automata). In~\cite{DBLP:journals/tocl/AlurETP01} it is also noted that the undecidability proof for LTL with parameterized equality modalities 
can be adapted to prove the undecidability of the satisfiability problem for LTL extended with parameterized \emph{upper bound modalities} of the form $\U_{\leq p}$ and
\emph{lower bound modalities} of the form $\U_{>p}$ unless %we disallow parameters to occur in constraints of both lower and upper bound modalities. 
we restrict every parameter to occur in \emph{either} lower bound modalities \emph{or} upper bound modalities, but not in both.

The restriction on the parameters of a parametric timed automaton to occur either as a lower bound or as an upper bound also forms an important subclass of parametric timed automata, called \emph{lower bound/upper bound (L/U) automata}~\cite{DBLP:journals/jlp/HuneRSV02}. For this subclass the emptiness problem is decidable independent of the number of parametric clocks, and for both finite~\cite{DBLP:journals/jlp/HuneRSV02} and infinite runs~\cite{DBLP:journals/fmsd/BozzelliT09}. Model checking L/U automata with parametric extensions of MITL~\cite{DBLP:journals/jacm/AlurFH96} in the  \emph{interval-based} semantics is decidable~\cite{DBLP:journals/fmsd/BozzelliT09,DBLP:conf/lata/GiampaoloTN10}. 
Recall that constraints occurring at modalities of MITL formulae are not allowed to be of the form $=n$ (not even if the constraint is \emph{concrete}, \ie, $n\in\N$); in fact, the satisfiability and model checking problems for (non-parametric) MTL in the interval-based semantics are undecidable~\cite{Henzinger}.

A crucial aspect of our undecidability proof is the fact that MTL formulae can be used to encode computations of \emph{channel machines with insertion errors}~\cite{DBLP:conf/fossacs/OuaknineW06}: 
For every channel machine $\cm$, there is an MTL formula $\varphi_\cm$ that is satisfiable if, and only if, $\mathcal{C}$ has a halting computation that may contain insertion errors. 
This fact was used in~\cite{DBLP:conf/fossacs/OuaknineW06} to prove the lower complexity bound of the satisfiability problem for MTL over finite timed words.
In our proof, we use the parameterized timed automaton to \emph{exclude} insertion errors in the timed words encoding computations of $\cm$. 
We remark that the idea for this proof is similar to the proof of the  undecidability for the model checking problem for one-counter machines and Freeze LTL with one register (LTL$^\downarrow_1$, for short) ~\cite{DBLP:conf/fossacs/DemriLS08}:
In~\cite{DBLP:journals/tocl/DemriL09}, it is proved that LTL$^\downarrow_1$ formulae can be used to encode halting computations of \emph{counter automata with incrementing errors}. Like MTL, LTL$^\downarrow_1$ is not capable to exclude such errors.
In~\cite{DBLP:conf/fossacs/DemriLS08}, it is shown that this incapability can be repaired by combining the formula with a non-deterministic one-counter machine.
Let us, however, note that there are substantial technical differences between the formalisms MTL and parametric timed automata on the one side, and LTL$^\downarrow_1$ and one-counter machines on the other side.

\section{Parametric Timed Automata}
We use $\N$, $\QP$, and  $\RP$ to denote the non-negative integers, non-negative rationals, and the non-negative reals, respectively.
In this section, we fix a finite alphabet $\Sigma$,
a finite set $\params=\{p_1,\dots,p_m\}$ of \emph{parameters}, and a finite set $\clocks=\{x_1,\dots,x_n\}$ of \emph{clocks}. 

We define {\em  clock constraints} $\phi$ over $\clocks$ and $\params$ to be conjunctions of formulae of the form $\clock\sim c$, where $\clock\in \clocks$, $c\in\N\cup\params$, and $\sim\in\{<,\leq,=,\geq,>\}$. 
We use $\Phi(\clocks,\params)$ to denote the set of all clock constraints over $\clocks$ and $\params$.
A \emph{clock valuation} is a function from $\clocks$ to $\RP$.
%We let $\nu_0$ be a special clock valuation assigning $0$ to each clock variable.
For $\delta\in\RP$, we define $\nu+\delta$ to be $(\nu+\delta)(\clock)=\nu(\clock)+\delta$ for each $\clock\in \clocks$.
For $\lambda\subseteq \clocks$, we define $\nu[\lambda:=0]$ by $(\nu[\lambda:=0])(x)=0$ if $x\in\lambda$, and otherwise $(\nu[\lambda:=0])(x)=\nu(x)$. 

A parameter valuation is a function $\pv:\params\to\QP$ assigning a non-negative rational to each parameter.

A clock valuation $\nu$ and a parameter valuation $\pv$ satisfy a clock constraint $\phi$, written $(\nu,\pv)\models\phi$, if the expression obtained from $\phi$ by replacing each parameter $p$ by $\pv(p)$ and each clock $x$ by $\nu(x)$ evaluates to true.

A \emph{parametric timed automaton} is a tuple $\A=(\Sigma, \locs, \locs_0, \clocks, \params, \edges, \locs_F)$, where
\begin{itemize}
\item $\locs$ is a finite set of locations,
\item $\locs_0\subseteq \locs$ is the set of \emph{initial} locations,
\item $\edges \subseteq \locs\times\Sigma\times\Phi(\clocks,\params)\times 2^\clocks\times\locs$ is a finite set of \emph{edges},
\item $\locs_F\subseteq\locs$ is the set of \emph{final} locations. 
\end{itemize}
Each edge $(\loc,a,\phi,\lambda,\loc')$ represents a discrete transition from $\loc$ to $\loc'$ on the input symbol $a$. The clock constraint $\phi$ specifies the bounds on the value of the clocks, and the set $\lambda$ specifies the clocks to be reset to zero.

A \emph{global state} of $\A$ is a pair $(\loc,\nu)$, where $\loc\in\locs$ represents the current location, and the clock valuation $\nu$ represents the current values of all clocks. 
The behaviour of $\A$ depends upon the current global state and the parameter valuation. 
Each parameter valuation $\pv$ induces a $(\Sigma,\RP)$-labelled transition relation $\tau_\pv$ over the set of all global states of $\A$ as follows: 
$\langle(\loc,\nu),(a,\delta),(\loc',\nu')\rangle\in\tau_\pv$, where $a\in\Sigma$ and $\delta\in\RP$, if, and only if, there is an edge $(\loc,a,\phi,\lambda,\loc')\in\edges$ such that for all clocks $x\in\clocks$ we have $(\nu(x)+\delta, \pv)\models\phi$, and $\nu'=(\nu(x)+\delta)[\lambda:=0]$.
A $\pv$-run of $\A$ is a finite sequence 
$\Pi_{1\leq i\leq k}\langle (\loc_{i-1},\nu_{i-1}),(a_i,\delta_i),(\loc_i,\nu_i)\rangle$ such that $\langle (\loc_{i-1},\nu_{i-1}),(a_i,\delta_i),(\loc_i,\nu_i)\rangle \in \tau_\pv$ 
%and  
%$(\loc_i,\nu_i)=(\loc_{i+1},\nu_{i+1})$ 
for every $i\in\{1,\dots,k\}$. 
%$\Pi_{1\leq i\leq k}\langle (\loc_i,\nu_i),(a_i,\delta_i),(\loc'_i,\nu'_i)\rangle$ satisfying $\langle (\loc_i,\nu_i),(a_i,\delta_i),(\loc'_i,\nu'_i)\rangle \in \tau_\pv$ and  $(\loc'_i,\nu'_i)=(\loc_{i+1},\nu_{i+1})$ for every $i\in\{1,\dots,k-1\}$. 
A $\pv$-run is \emph{successful} if $\loc_0 \in \locs_0$, $\nu_0(x)=0$,  and $\loc_k\in\locs_F$. 

A {\em timed word} is a non-empty 
finite sequence $(a_1,t_1)\dots(a_k,t_n)\in(\Sigma\times\RP)^+$ such that the sequence $t_1,\dots,t_n$ of timestamps is non-decreasing.
We say that a timed word is \emph{strictly monotonic} if $t_1,\dots,t_n$ is strictly increasing. 
We use $\TW$ to denote the set of finite timed words over $\Sigma$.
A set $L\subseteq \TW$ is called a {\em timed language}.

Given a parametric timed automaton $\A$ and a parameter valuation $\pv$, we associate with each $\pv$-run $\Pi_{1\leq i\leq k}\langle (\loc_{i-1},\nu_{i-1}),(a_i,\delta_i),(\loc_i,\nu_i)\rangle$ the timed word $(a_1,\delta_1)(a_2,\delta_1+\delta_2)\dots (a_k,\sum_{1\leq i\leq k}\delta_k)$. 
We define $L_\pv(\A)$ to be the set of timed words $w$ for which there is a successful $\pv$-run of $\A$ that is associated with $w$.
A parameter valuation $\pv$ is \emph{consistent with $\A$} if $L_\pv(\A)$ is not empty. 
We use $\Pi(\A)$ to denote the set of parameter valuations that are consistent with $\A$. 

We say that a parametric timed automaton $\A$ is \emph{deterministic} 
if $\locs_0$ is a singleton, and whenever ${(\loc,a,\phi_1,\lambda_1,\loc_1)}$ and ${(\loc,a,\phi_2,\lambda_2,\loc_2)}$ are two different edges in $\A$, then for all parameter valuations $\pv$ and clock valuations $\nu$ we have ${(\nu,\pv)\not\models \phi_1\wedge\phi_2}$.

\begin{example}
	Figure \ref{figure_ex} shows a parametric timed automaton over the alphabet $\Sigma=\{a,b\}$ using a parametric clock $x$ and a clock $y$, and one parameter $p$. 
	Assume $\pv(p)=n^{-1}$ for some $n\in\N$. 
	Then $L_\pv(\A)$ contains a single timed word, namely $(a,\pv(p))(a,2\pv(p))\dots (a,n\pv(p))(b,(n+1)\pv(p))\dots(b,2n\pv(p))$. 
	For all other parameter valuations $\pv$, $L_\pv(\A)=\emptyset$, \ie, they are not consistent with $\A$. 
	Hence we have $\Pi(\A) = \{\pv\mid \pv(p)=n^{-1} \text{ for some } n\in\N\}$. 
	Note that $\A$ is not deterministic, but it can be made deterministic by adding the clock constraint $y<1$ to the loops in locations $1$ and $2$.
\end{example}

\begin{figure}
\begin{center}
\begin{picture}(80,18)(0,-18)
%\put(0,-18){\framebox(110,18){}}
\node[NLangle=0.0,Nmarks=i,ilength=3,Nw=4.0,Nh=4.0,Nmr=2.0](n0)(5.0,-14.0){$1$}
\node[NLangle=0.0,Nw=4.0,Nh=4.0,Nmr=2.0](n1)(40.0,-14){$2$}
\node[NLangle=0.0,Nmarks=f,flength=3,Nw=4.0,Nh=4.0,Nmr=2.0](n2)(75.0,-14){$3$}
\drawloop[loopdiam=4](n0){}
\put(1,-3.5){\footnotesize{$a,x=p$}}
\put(2,-7){\footnotesize{$x:=0$}}
\drawedge[curvedepth=4.0](n0,n1){\footnotesize{$a, x=p,y=1$}}
\put(16,-14){\footnotesize{$x,y:=0$}}
\put(36,-3.5){\footnotesize{$b,x=p$}}
\put(37,-7){\footnotesize{$x:=0$}}
\drawloop[loopdiam=4](n1){}
\drawedge[curvedepth=4.0](n1,n2){\footnotesize{$b, x=p,y=1$}}
\end{picture}
\caption{A parametric timed automaton $\A$.}
\label{figure_ex}
\end{center}
\end{figure}
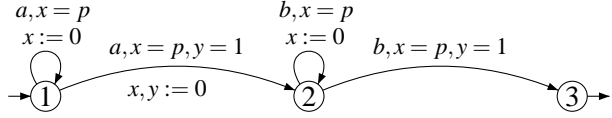

\section{Metric Temporal Logic}

The set of MTL formulae is built up from $\Sigma$ 
        by boolean connectives and a constraining version of the {\em until}
modality:
        $$\varphi \ndefeq a \sep \neg\varphi \sep \varphi_1\wedge\varphi_2 \sep
        \varphi_1\U_{I}\varphi_2  $$
        where $a\in\Sigma$ and $I \subseteq \RP$ is an open, closed, or half-open interval with endpoints in $\N\cup\{\infty\}$.
        Note that we do \emph{not} allow parameters as endpoints.
If $I=\RP$, then we may omit the annotation $I$ on $\U_I$.

We interprete MTL formulae in the \emph{pointwise semantics}, \ie, over finite timed words over $\Sigma$. 
Let $w=(a_1,t_1)(a_2,t_2)\dots(a_n,t_n)$ be a timed word, and let  $i\in\{1,\dots,n\}$. 
We define the {\em satisfaction relation for MTL}, denoted by $\models$,
inductively as follows:
\begin{flalign*}
	(w,i) \models a &\hspace{2mm}\Leftrightarrow\hspace{2mm} a_i=a\\
	 (w,i) \models
\neg\varphi & \hspace{2mm}\Leftrightarrow\hspace{2mm} (w,i)\not\models \varphi, \\
	 (w,i) \models \varphi_1\wedge\varphi_2 & \hspace{2mm}\Leftrightarrow\hspace{2mm}  
	(w,i)\models \varphi_1 \textrm{ and } (w,i)\models \varphi_2,\\
	 (w,i)\models \varphi_1\U_{I}\varphi_2 & \hspace{2mm}\Leftrightarrow\hspace{2mm} \exists j.
i<j\leq |w|: (w,j)\models\varphi_2 \textrm{ and } t_j - t_i \in I, \textrm{
and } \forall k.i< k<j:(w,k)\models\varphi_1.
\end{flalign*}
We say that a timed word $w\in\TW$ satisfies an MTL formula $\varphi$, written $w\models\varphi$, if $(w,1)\models \varphi$.
Given an MTL formula $\varphi$, we define $L(\varphi)\defeq \{w\in T\Sigma^+\mid w\models \varphi\}$. 
We use the following syntactical abbreviations:
$\varphi_1\vee\varphi_2 \defeq\neg(\neg\varphi_1\wedge\neg\varphi_2)$,
$\varphi_1\rightarrow \varphi_2 \defeq \neg\varphi_1\vee\varphi_2$, 
$\true \defeq p \vee \neg p$, 
$\false \defeq \neg\true$,
$\X_I\varphi \defeq \false \U_I\varphi$, 
$\finally_I\varphi\defeq\true\U_I\varphi$,
$\glob_I\varphi \defeq\neg\finally_I\neg\varphi$.
Observe that the use of the \emph{strict} semantics for the until modality is
essential to derive the next modality. 

\problemx{MTL-Model Checking Problem for Parametric Timed Automata}{A parametric timed automaton $\A$, an MTL formula $\varphi$.}{Is there some parameter valuation $\pv$ such that for every $w\in L_\pv(\A)$ we have $w\models\varphi$?}

In general, the MTL-model checking problem is undecidable for parametric timed automata. 
This follows from the undecidability of the emptiness problem for parametric timed automata with three or more parametric clocks~\cite{DBLP:conf/stoc/AlurHV93}.
%Recall that the  reachability problem for  parametric timed automata with at most two clocks is still open, more than 20 years after the introduction of parametric timed automata in~\cite{DBLP:conf/stoc/AlurHV93}.
In the next section, 
we prove the undecidability of the MTL-model checking problem for parametric timed automata using one parametric clock and one parameter.

\section{Main Result}
\begin{theorem}
	\label{thm_main}
	The MTL-model checking problem for parametric timed automata is undecidable, even if the automaton uses only one clock and one parameter and is deterministic.
\end{theorem}
The remainder of this section is devoted to the proof of Theorem \ref{thm_main}. 
The proof is a reduction of the control state reachability problem for channel machines, which we introduce in the following. 

\subsection{Channel Machines}
Let $\Gamma$ be a finite alphabet. 
We use $\varepsilon$ to denote the \emph{empty word} over $\Gamma$.
Given two finite words $x,y\in \Gamma^*$, we use $x\cdot y$ to denote the \emph{concatenation} of $x$ any $y$. 
We define the order $\leq$ over the set of finite words over $\Gamma$ by 
$x_1 x_2 \dots x_m \leq y_1 y_2 \dots y_n$ if there exists a strictly increasing function $f:\{1,\dots,m\}\to\{1,\dots,n\}$ such that $x_i = y_{f(i)}$ for every $i\in\{1,\dots,m\}$. 

A \emph{channel machine} consists of a finite-state automaton acting on an  unbounded fifo channel. Formally, 
a channel machine is a tuple $\cm=(S,s_I,M,\Delta)$, where
\begin{itemize}
\item $S$ is a finite set of \emph{control states},
	\item $s_I\in S$ is the initial control state,
	\item $M$ is a finite set of \emph{messages},
	\item $\Delta \subseteq S\times L\times S$ is the transition relation over the label set $L=\{m!,m?\mid  m\in M\}\cup\{\varepsilon\}$.
\end{itemize}
A \emph{configuration} of $\cm$ is a tuple $(s,x)$, where $s\in S$ is the control state and $x\in M^*$ represents the contents of the channel.
The rules in $\Delta$ induce an $L$-labelled transition relation $\to$ over the set of configurations of $\cm$ as follows:
\begin{itemize}
\item $\langle (s,x),m!,(s',x')\rangle\in\to$ if, and only if, there exists some transition $(s,m!,s')\in\Delta$, $x\in\Sigma^*$, and $x'=x\cdot m$, \ie, $m$ is added to the tail of the channel.
\item $\langle (s, x),m?,(s',x')\rangle\in\to$ if, and only if, there exists some transition $(s,m?,s')\in\Delta$, $x'\in\Sigma^*$, and $x=m\cdot x'$, \ie, $m$ is the head of the current channel content. 
\item $\langle (s,x),\varepsilon,(s',x')\rangle\in\to$ if, and only if, there exists some transition $(s,\varepsilon,s')\in\Delta$ and $x=\varepsilon$, \ie, the channel is empty, and $x'=x$. 
\end{itemize}
Next, we define another $L$-labelled transition relation $\leadsto$  over the set of configurations of $\cm$. The relation $\leadsto$ is a superset of $\to$. It contains some additional transitions which result from \emph{insertion errors}. We define $\langle (s,x_1),l,(s,x'_1)\rangle\in\leadsto$, if, and only if, $\langle (s,x), l,(s',x')\rangle\in\to$, $x_1\leq x$, and $x'\leq x'_1$.
A computation of $\cm$ is a finite sequence $\Pi_{1\leq i\leq k} \langle (s_{i-1},x_{i-1}),l_i,(s_i,x_i)\rangle$ such that $\langle (s_{i-1},x_{i-1}),l_i,(s_i,x_i)\rangle\in\leadsto$ for every $i\in\{1,\dots,k\}$. 
We say that a computation is \emph{error-free} if for all $i\in\{1,\dots,k\}$ we have $\langle (s_{i-1},x_{i-1}),l_i,(s_i,x_i)\rangle\in\to$.
Otherwise, we say that the computation is \emph{faulty}.
\problemx{Control State Reachability Problem for Channel Machines}{A channel machine $\cm$ with control states $S$, a control state $s_F\in S$.}{Is there an error-free computation of $\cm$ from $(s_I,\varepsilon)$ to $(s_F,x)$ for some $x\in M^*$?}
The control state reachability problem is undecidable for channel machines, because channel machines are Turing-powerful~\cite{Brand:1983:CFM:322374.322380,DBLP:journals/fuin/AbdullaDOQW08}.

\subsection{Encoding Faulty Computations}
For the remainder of this section, let $\cm=(S,s_I,M,\Delta)$ be a channel machine and let $s_F\in S$.
We construct an MTL formula $\varphi_\cm$ that is satisfiable if, and only if,
there exists some $x\in M^*$ such that $\cm$ has a computation from $(s_I,\varepsilon)$ to $(s_F,x)$ that may be faulty. 
Later we are going to define a parametric timed automaton $\A_\cm$ with one clock and one parameter to exclude faulty computations from $L(\varphi_\cm)$.

Let $\Sigma = S\cup M\cup L\cup \{\#,\star\}$, where $\#$ and $\star$ do not occur in $S\cup M\cup L$. 
We start with defining a timed language $L(\cm)$ over $\Sigma$ that consists of all timed words that encode (potentially faulty) computations of $\cm$ from $(s_I,\varepsilon)$ to $(s_F,x)$ for some $x\in M^*$. 
The definition of $L(\cm)$ follows the ideas presented in~\cite{DBLP:conf/fossacs/OuaknineW06}.
Let $\gamma\defeq \Pi_{1\leq i\leq k} \langle (s_{i-1},x_{i-1}),l_i,(s_i,x_i)\rangle$ be a computation of $\cm$ with $s_0=s_I$, $x_0=\varepsilon$, and $s_k=s_F$.
Each configuration $(s_i,x_i)$ occurring in $\gamma$ is encoded by a timed word of duration one starting with $s_0$ at time $\delta$ for some arbitrary $\delta\in\RP$. 
Every symbol $s_i$ is followed by $l_{i+1}$ after one time unit, and by $s_{i+1}$ after two time units. 
The content $x_i$ of the channel is stored in the time interval between $s_i$ and $l_{i+1}$. 
Note that due to the denseness of the time domain we can indeed store the channel content without any restriction on its length. 
An important detail of the definition of $L(\cm)$ is that for every message symbol $m$ between $s_i$ and $l_{i+1}$, there is a copy in the encoding of the next configuration exactly two time units later, unless the label of the current transition is $m?$. In that case, the symbol $m$ is simply removed from the encoding of the configuration.

For our reduction to work, we have to change the idea in some details. 
First, we define a timed language $L(\cm,n)$ for every $n\in \N$, where $n$ is non-deterministically chosen and is supposed to represent the expected maximum length of the channel content during a computation. The empty channel in the initial configuration will be represented by a timed word with $n$ hash symbols between $s_0$ and $l_1$. 
Second, we put a stronger condition on the copy policy of the messages. 
We require that for 
every hash symbol between $s_0$ and $l_1$ there is a message or hash symbol with \emph{the same fractional part} between $s_i$ and $l_{i+1}$ for every $i\in\{1,\dots,k-1\}$. 
In Fig. 2, we present some examples to explain the details. 
(a) If the current instruction is of the form $m_1!$ for some $m_1\in M$, then in the encoding of the next configuration, the first hash symbol between the control state symbol and the next label symbol is replaced by $m_1$.
(b) If in the encoding of the current configuration there is no hash symbol left, \ie, the expected maximum length of the channel content is exceeded, then  a new symbol $m_1$ is inserted at the end of the encoding of the next configuration. The timestamp of the newly inserted event can be any time strictly between the timestamps of the last message symbol and the next label symbol.
(c) If the current instruction is of the form $m_1?$ and the first symbol in the encoding of the current configuration is $m_1$, then we replace $m_1$ by a new hash symbol at the end of the encoding of the next configuration, and additionally shift the fractional parts of the timestamps of the copies of all remaining symbols for one position to the right. 
(d) If the first symbol is not $m_1$, \ie, an insertion error is occurring, then we  insert a new hash symbol at the end of the encoding of the next configuration. Next, we give the formal definition of $L(\cm,n)$.
\begin{figure}
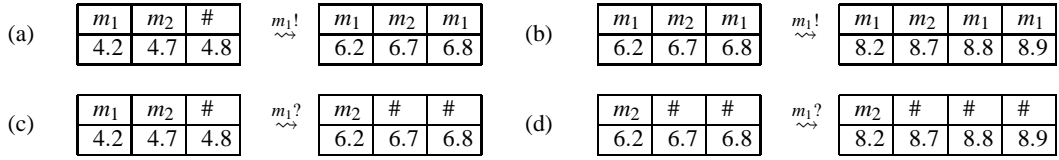

\begin{center}
	\footnotesize{
		\begin{tabular}{p{5mm}lp{2mm}lp{7mm}lp{2mm}l}
			(a) & \begin{tabular}{|p{3mm}|p{3mm}|p{3mm}|}
		\hline
		
		$m_1$ &  $m_2$ & $\#$ \\
		\hline
		$4.2$ &  $4.7$ & $4.8$ \\
		\hline
	\end{tabular} & $\buildrel m_1! \over \leadsto$ & \begin{tabular}{|p{3mm}|p{3mm}|p{3mm}|}
		\hline
		
		$m_1$  & $m_2$ & $m_1$ \\
		\hline
		$6.2$ &  $6.7$ & $6.8$ \\
		\hline
	\end{tabular} & \ \ (b) & \begin{tabular}{|p{3mm}|p{3mm}|p{3mm}|}
		\hline
		
		$m_1$  & $m_2$ & $m_1$ \\
		\hline
		$6.2$ &  $6.7$ & $6.8$ \\
		\hline
	\end{tabular} & $\buildrel m_1! \over \leadsto$ & \begin{tabular}{|p{3mm}|p{3mm}|p{3mm}|p{3mm} |}
		\hline
		
		$m_1$  & $m_2$ & $m_1$ & $m_1$ \\
		\hline
		$8.2$ & $8.7$ & $8.8$ & $8.9$\\
		\hline
	\end{tabular} \\
	& & &\\
	(c) & \begin{tabular}{|p{3mm}|p{3mm}|p{3mm}|}
		\hline
		
		$m_1$ &  $m_2$ & $\#$ \\
		\hline
		$4.2$ &  $4.7$ & $4.8$ \\
		\hline
	\end{tabular} & $\buildrel m_1? \over \leadsto$ & \begin{tabular}{|p{3mm}|p{3mm}|p{3mm}|}
		\hline
		
		$m_2$  & $\#$ & $\#$ \\
		\hline
		$6.2$ &  $6.7$ & $6.8$ \\
		\hline
	\end{tabular} & \ \ (d) & \begin{tabular}{|p{3mm}|p{3mm}|p{3mm}|}
		\hline
		
		$m_2$  & $\#$ & $\#$ \\
		\hline
		$6.2$ &  $6.7$ & $6.8$ \\
		\hline
	\end{tabular} & $\buildrel m_1? \over \leadsto$ & \begin{tabular}{|p{3mm}|p{3mm}|p{3mm}|p{3mm} |}
		\hline
		
		$m_2$  & $\#$ & $\#$ & $\#$ \\
		\hline
		$8.2$ & $8.7$ & $8.8$ & $8.9$\\
		\hline
	\end{tabular}
\end{tabular}
	}
	
\label{figure_copy}
\end{center}
\caption{Encoding of the channel content}
\end{figure}
Let $n\in \N$. The timed language $L(\cm,n)$ consists of all timed words $w$ over $\Sigma$ that satisfy the following conditions:
\begin{itemize}
\item $w$ must be strictly monotonic. 
	
\item In $w$, every control state symbol $s$ different from $s_F$ is followed  by a label symbol $l$ after one time unit, and by a control state symbol $s'$ after two time units, provided that $(s,l,s')\in \Delta$. The symbol $s_F$  is followed by $\star$ after one time unit. Control state symbols, label symbols and the symbol $\star$ must not occur anywhere else in $w$.  
	
%\item $\glob\langle \bigvee S \rightarrow ((\glob_{<2}\neg\bigvee S) \wedge (\glob_{(0,1)\cup(1,2)}\neg\bigvee L)\rangle$. State and label symbols must not occur anywhere else in a model.   
	
\item Symbols in $M\cup\{\#\}$ may occur in $w$ between a control state symbol and a label symbol. They may not occur anywhere else in $w$. 
	
\item Between a control state and a label symbol, hash symbols $\#$ may only occur after message symbols $m\in M$.
	
\item The (untimed) prefix of $w$ must be of the form $s_I \#^n l s$ for some $l\in L,s\in S$.
	
\item $w$ must contain $s_F$. 
\end{itemize}
Assume that $w$ contains the infix $(s,\delta)(\sigma_1,\delta+\delta_1)(\sigma_2,\delta+\delta_2)\dots(\sigma_m,\delta+\delta_m)(l,\delta+1)$ for some $s\in S\backslash\{s_F\}$, $l\in L$, $\delta\in\RP$ and $0<\delta_1<\delta_2<\dots<\delta_m<1$.
\begin{itemize}
	
\item If $l=\varepsilon$,
	then $\sigma_i=\#$ for all $i\in\{1,\dots,m\}$ (\ie, the channel is indeed empty), and for each $\sigma_i$ there is a copy two time units later. 
	
\item If $l=m!$,
	then we distinguish between two cases:
	If there is some $i\in\{1,\dots,m\}$ such that $\sigma_i=\#$, 
	then \emph{replace} $\sigma_j$  by $m$ two time units later, where $j\in\{1,\dots,m\}$ is the smallest number such that $\sigma_j=\#$. For each $k\in\{1,\dots,m\}\backslash\{j\}$, there is a copy of $\sigma_k$ two time units later. 
	Otherwise, \ie, if for all $i\in\{1,\dots,m\}$ we have $\sigma_i\neq\#$,
	then for each $i\in\{1,\dots,m\}$, there is a copy of $\sigma_i$ two time units later. Further, a new symbol $m$ is added between the copy of $\sigma_m$ and the following symbol in $L\cup\{\star\}$. 
	Note that this corresponds to the case where $n$ has been chosen too small to capture the maximum length of the channel content during the computation.

\item If $l=m?$, then we distinguish between two cases:
	If $\sigma_1 = m$, 
	then for each $i\in\{2,\dots,m\}$,
	there is a copy of $\sigma_i$ two time units after the occurrence of $\sigma_{i-1}$. 
	Further there is a new hash symbol two time units after the occurrence of $\sigma_m$. 
	Otherwise, \ie, if $\sigma_1\neq m$, then there is a copy of $\sigma_i$ two time units later for every $i\in\{1,\dots,m\}$. 
	Further, the encoding of the next configuration contains an additional hash symbol 
	between the copy of $\sigma_m$ and the next symbol in $L\cup\{\star\}$. Note that this case corresponds to an \emph{insertion error}. 
\end{itemize}
Let $w_1=(a_1,t_1)\dots(a_k,t_k)$ and $w_2=(a'_1,t'_1)\dots(a'_{k'},t'_{k'})$ be two timed words.
If $t_k\leq t'_{1}$, then we define the \emph{concatenation} of $w_1$ and $w_2$, denoted by $w_1\cdot w_2$, to be the timed word  $(a_1,t_1)\dots(a_k,t_k)(a'_1,t'_1)\dots(a'_{k'},t'_{k'})$. 
Let $w\in L(\cm,n)$. 
We use $\max(w)$ to denote the maximum number of symbols in $M\cup\{\#\}$ that occur in $w$ between a control state symbol and a symbol in $L\cup\{\star\}$.
Clearly, every timed word in $L(\cm,n)$ is of the form 
$$ (s_0,\delta)\cdot w_1\cdot (l_1,\delta+1)(s_1,\delta+2)\cdot w_2\cdot(l_2,\delta+3)\dots(s_F,\delta+N)\cdot w_N \cdot (\star,\delta+N+1)$$
for some $\delta\in \RP$ and $N\in\N$, where $s_0=s_I$ and for every $i\in\{1,\dots,N\}$,  $w_i$ is of the form 
$$w_i=(\sigma^i_1,\delta+2(i-1)+\delta^i_1)(\sigma^i_2,\delta+2(i-1)+\delta^i_2)\dots (\sigma^i_{n_i},\delta+2(i-1)+\delta^i_{n_i})$$ for some $n_i\in \N$ with $n_1=n$,  and $0<\delta^i_1<\delta^i_2<\dots<\delta^i_{n_i}<1$.	In the following, whenever we refer to a timed word $w\in L(\cm,n)$, we assume that $w$ is of this form. 
The next lemma states that the fractional parts of the initial time delays $\delta^1_1,\dots,\delta^1_{n_1}$ are not lost. This will be important later. 
\begin{lemma}
	\label{lemma_function}
	Let $n\in \N$ and let $w\in L(\cm,n)$.
	For every $i\in\{1,\dots,N-1\}$
	there exists a strictly increasing function $f_i:\{1,\dots,n_i\}\to\{1,\dots,n_{i+1}\}$ such that $\delta^{i}_j=\delta^{i+1}_{f_i(j)} $ for every $j\in\{1,\dots,n_i\}$. 
\end{lemma}
\begin{proof} 
	The proof is by induction on $N$.
	(Induction base:) 
	Observe that $\sigma^1_i=\#$ for every $i\in\{1,\dots,n_1\}$.
	Assume $l_1=\varepsilon$.
	Then for every $j\in\{1,\dots,n_1\}$, there is a copy of $\sigma^1_j$ two time units later.
	If $l_1 = m!$,
	then for every 
	$j\in\{2,\dots,n_1\}$, there is a copy of $\sigma^1_j$ two time units later, and $\sigma_1^1$ is replaced by $m$ two time units later. 
	If $l_1=m?$, 
	then for every $j\in\{1,\dots,n_1\}$, there is a copy of $\sigma^1_j$ two time units later, and there is an additional symbol $\#$ between the copy of $\sigma^1_{n_1}$ and $l_2$. 
	Whatever case, the definition of $L(\cm,n)$ does not exclude that new symbols in $M\cup \{\#\}$ are inserted somewhere between $s_1$ and $l_2$. 
	Thus we have $n_1\leq n_2$. 
	Moreover, since there is a copy for each symbol two time units later, 
	there exists a strictly increasing function $f:\{1,\dots,n_1\}\to\{1,\dots,n_2\}$ such that $ \delta^1_j= \delta^2_{f(j)}$ for every $j\in\{1,\dots,n_1\}$. 
	(Induction step)
	Assume that the claim holds for all $i\in\{1,\dots,k\}$. 
	We prove it also holds for $k+1$. 
	We only treat the two remaining cases.
	First, assume $l_{k+1}=m?$ and $\sigma_1^{k+1}=m$. 
	By definition, for every $j\in\{2,\dots,n_{k+1}\}$, 
	there is a copy of $\sigma_j^{k+1}$  two time units after the occurrence of symbol $\sigma_{j-1}^{k+1}$.
	Further, the first symbol $m$ is replaced by a new hash symbol two time units after the occurrence of $\sigma_{n_{k+1}}^{k+1}$. 
	Second, assume $l_{k+1}=m!$ and we have $\sigma_j^{k+1}\neq\#$ for every $j\in\{1,\dots,n_{k+1}\}$. Then, for each $j\in\{1,\dots,n_{k+1}\}$, there is a copy of $\sigma^{k+1}_j$ two time units later, and a new symbol $m$ is added after the copy of $\sigma_{n_{k+1}}^i$. 
	Whatever case, the definition of $L(\cm,n)$ does not exclude that new symbols in $M\cup\{\#\}$ are inserted between $s_{k+2}$ and $l_{k+2}$.
	Hence $n_{k+1}\leq n_{k+2}$. 
	Since for every $j\in\{1,\dots,n_{k+1}\}$ the symbol $\sigma_j^{k+1}$ is copied or replaced two time units later, there exists a strictly increasing function $f:\{1,\dots,n_{k+1}\}\to\{1,\dots,n_{k+2}\}$ such that $\delta^{k+1}_j = \delta^{k+2}_{f_{k+1}(j)}$ for every $j\in\{1,\dots,n_{k+1}\}$. 
\end{proof}
Let $\gamma\defeq \Pi_{1\leq i\leq k} \langle (s_{i-1},x_{i-1}),l_i,(s_i,x_i)\rangle$ be a finite computation of $\cm$. We use $\max(\gamma)$ to denote the maximum length of the channel content occurring in $\gamma$, formally: $\max(\gamma)\defeq\max\{|x_i| \mid 0\leq x_i\leq k\}$. 
\begin{lemma}
		\label{lemma_comp_to_word}
		For each error-free computation $\gamma$ of $\cm$ from $(s_I,\varepsilon)$ to $(s_F,x)$ for some $x\in M^*$,
		and every $\delta\in\RP$, $0<\delta_1<\delta_2<\dots<\delta_{\max(\gamma)}<1$, 
		there exists some timed word 
		$w\in L(\cm,\max(\gamma))$ such that the prefix of $w$ is of the form $(s_I,\delta)(\#,\delta+\delta_1)\dots(\#,\delta+\delta_{\max(\gamma)})(l_1,\delta+1)$ for some $l_1\in L$,  and $\max(w)=\max(\gamma)$. 
	\end{lemma}
	\begin{proof}
		Let $\gamma$ be an errror-free computation of $\cm$ of the form
		$\Pi_{1\leq i\leq k} \langle (s_{i-1},x_{i-1}),l_i,(s_i,x_i)\rangle$ 
		where $s_0=s_I$, $x_0=\varepsilon$ and $s_k= s_F$. Further let $n=\max(\gamma)$. 
		Now assume $\delta\in\RP$ and $0<\delta_1<\delta_2<\dots<\delta_n<1$.  
		Clearly there is some $w\in L(\cm,n)$ whose prefix is of the form
		$u_1 = (s_I,\delta)(\#,\delta+\delta_1)\dots(\#,\delta+\delta_n) (l_1,\delta+1)$.
		We prove that there exists some $w\in L(\cm,n)$ such that $u_1$ is the prefix of $w$ and $\max(w)=n$, $\ie$,  for every $i\in\{1,\dots,k\}$, the number of symbols in $M\cup\{\#\}$ between  $s_{i-1}$ and  $l_{i}$  (and between $s_k$ and $\star$) is equal to $n$. The proof is by induction on $k$.
	
		(Induction base:) 
		Assume $l_1 = \varepsilon$. 
		By definition,
		there must be a copy for each $\#$ exactly two time units later. 
		The addition of  new symbols is not required.
		If $l_1=m!$,
		then by definition the first occurrence of $\#$ is \emph{replaced} by $m$ exactly two time units later, and for each of the remaining $\#$ there is a copy two time units later. The addition of  new symbols is not required. 
		Note that the case $m?$ cannot occur because $\gamma$ is error-free. 
		Hence,
		there exists some timed word $w\in L(\cm,n)$ whose prefix is of the form $u_1 \cdot u_2$, where
		$u_2 = (s_1,2+\delta) (\sigma_1^2,2+\delta+\delta_1)(\#,2+\delta+\delta_2)\dots(\#,2+\delta+\delta_n) (l_2,2+\delta+1)$ for some $\sigma_1^2\in M \cup \{\#\}$.

		(Induction step:)
		Assume there is some timed word $w\in L(\cm,n)$ whose prefix is of the form $u_1 \cdot \dots \cdot u_p$ for some $p<k$, 
		where 
		for every $i\in \{1,\dots,p\}$, $u_i$ is of the form
		$$(s_{i-1},2(i-1)+\delta)(\sigma_1^i,2(i-1)+\delta+\delta_1)(\sigma_2^i,2(i-1)+\delta+\delta_2)\dots(\sigma_n^i,2(i-1)+\delta+\delta_n)(l_i,2(i-1)+\delta+1)$$
		for some $\sigma_1^i,\dots,\sigma_n^i\in M\cup\{\#\}$. 
		
		Assume $l_p = m?$ for some $m\in M$. 
		By the fact that $\gamma$ is error-free,
		we know $\sigma_1^p = m$.
		By definition,
		there is a copy of $\sigma_i^p$ two time units after the occurrence of $\sigma^p_{i-1}$ for every $i\in\{2,\dots,n\}$,
		and there is a new hash symbol inserted two time units after the occurrence of $\sigma^p_{n}$. 		
		The addition of  new symbols is not required.  
		
		Assume $l_p = m!$ for some $m\in M$. 
		Recall that $n=\max(\gamma)$ is the maximum length of the channel content in $\gamma$. 
		Hence there must be some $j\in\{1,\dots,n\}$ such that $\sigma_j^p=\#$. 
		By definition,  
		the smallest $j\in\{1,\dots,n\}$ with $\sigma_j^p=\#$ is replaced by $m$ exactly two time units later. For each of the remaining symbols there is a copy two time units later. The addition of  new symbols is not required.
		
		Assume $l_p=\varepsilon$.  We can proceed as above, concluding that the addition of new symbols is not required. 		
		
		Hence,
		there exists some timed word $w\in L(\cm,n)$ whose prefix is of the form $u_1 \cdot u_2 \cdot\dots u_p\cdot u_{p+1}$, where  $u_{p+1}=(s_{p},2p+\delta)(\sigma_1^{p+1},2p+\delta+\delta_1)(\sigma_2^{p+1},2p+\delta+\delta_2)\dots(\sigma_n^{p+1},2p+\delta+\delta_n)(l_{p+1},2p+\delta+1)$
		for some $\sigma_1^{p+1},\dots,\sigma_n^{p+1}\in M\cup\{\#\}$. 
		
		We thus have proved that there exists some $w\in L(\cm, n)$ with $\max(w)=n$.	
	\end{proof}

	\begin{lemma}
		\label{lemma_word_to_comp}
		For each $n\in \N$ and $w\in L(\cm, n)$ with $\max(w)=n$,
		there exists some error-free computation $\gamma$ of $\cm$ from $(s_I,\varepsilon)$ to $(s_F,x)$ for some $x\in M^*$ with $\max(\gamma)\leq n$. 
	\end{lemma}
	\begin{proof}
		Let $n\in \N$ and let $w\in L(\cm,n)$ such that $\max(w)=n$. 
		Hence the number of symbols in $M\cup\{\#\}$ between every control state symbol and the following label symbol (or the symbol $\star$ if the state symbol is $s_F$) in $w$ is constantly equal to $n$. 
		This implies that (1) 
		whenever a control state symbol $s$ is followed by a label symbol $m?$ one time unit later, 
		then the next symbol after $s$ must be $m$, which will be replaced by a new hash symbol; (2) whenever a state symbol $s$ is followed  by a label symbol $m!$ one time unit later, then there must exist some hash symbol in between, and the first such hash symbol will be replaced by $m$; and (3) $w$ does not contain any spontaneously inserted symbols. 
		From (1) and (3) we can conclude that $w$ encodes an error-free computation. 
		From (2) we can conclude that the choice of $n$ is big enough to capture the maximum length of the channel content. 
		Hence there exists some error-free computation of $\cm$ from $(s_I,\varepsilon)$ to $(s_F,x)$ for some $x\in M^*$ with $\max(\gamma)\leq n$. 
	\end{proof}
	
	\subsection{Excluding Faulty Computations}
	Next we define a parametric timed automaton $\A_\cm$ over $\Sigma_\cm$ such that $L(\cm,n)\cap L(\A_\cm)$ consists of all timed words that encode \emph{error-free} computations of $\cm$ from $(s_I,\varepsilon)$ to $(s_F,x)$ for some $x\in M^*$.
The parametric timed automaton $\A_\cm$ is shown in Fig. \ref{figure_A}.
It uses one clock $x$, parametrically constrained by a single parameter $p$. 
Note that $\A_\cm$ is deterministic. 
		\begin{figure}
\begin{center}
		\begin{picture}(100,18)(0,-18)
%\put(0,-18){\framebox(110,18){}}
\node[NLangle=0.0,Nmarks=i,ilength=3,Nw=4.0,Nh=4.0,Nmr=2.0](n0)(5.0,-14.0){$1$}
\node[NLangle=0.0,Nw=4.0,Nh=4.0,Nmr=2.0](n1)(23.0,-14){$2$}
\node[NLangle=0.0,Nw=4.0,Nh=4.0,Nmr=2.0](n2)(50.0,-14){$3$}
\drawedge[curvedepth=4.0](n0,n1){\footnotesize{$s_I$}}
\put(9,-14){\footnotesize{$x:=0$}}
\drawloop[loopdiam=4](n1){}
\put(18,-4){\footnotesize{$\#,x=p$}}
\put(19,-7){\footnotesize{$x:=0$}}
\drawedge[curvedepth=4.0](n1,n2){\footnotesize{$L, x=p$}}
%\put(33,-14){\footnotesize{$x:=0$}}
\drawloop[loopdiam=4](n2){\footnotesize{$\Sigma\backslash\{s_F\}$}}
\node[NLangle=0.0,Nw=4.0,Nh=4.0,Nmr=2.0](n3)(70.0,-14){$4$}
\node[NLangle=0.0,Nmarks=f,flength=3,Nw=4.0,Nh=4.0,Nmr=2.0](n4)(95,-14){$5$}
\drawedge[curvedepth=4.0](n2,n3){\footnotesize{$s_F$}}
\put(56,-14){\footnotesize{$x:=0$}}
\drawloop[loopdiam=4](n3){}
\put(65,-4){\footnotesize{$M,\#, x=p$}}
\put(66,-7){\footnotesize{$x:=0$}}
\drawedge[curvedepth=4.0](n3,n4){\footnotesize{$\star,x=p$}}
\end{picture}
\caption{The parametric timed automaton $\A_\cm$ that excludes insertion errors. }
\label{figure_A}
\end{center}
\end{figure}
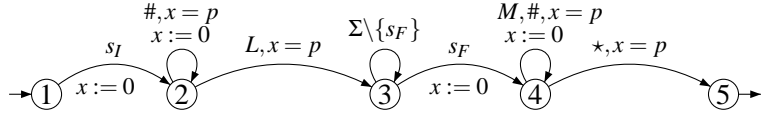
\begin{theorem}
	$\cm$ has an error-free computation from $(s_0,\varepsilon)$ to $(s_F,x)$ for some $x\in M^*$, if, and only if, 
	there exist $n\in \N$ and a parameter valuation $\pv$ such that $L(\cm,n)\cap L_\pv(\A_\cm)\neq \emptyset$.
\end{theorem}
\begin{proof}
	For the direction from left to right, let $\gamma\defeq\Pi_{1\leq i\leq k} \langle(s_{i-1},x_{i-1}),l_i,(s_i,x_i)\rangle$ 
	be an error-free computation of $\cm$ such that $s_0=s_I$, $x_0=\varepsilon$ and $s_k=s_F$. 
	Define $n = \max(\gamma)$.
	Let $\delta\in\RP$, 
	and define $\delta_i= {i\over{(n+1)}}$ for every $i\in\{1,\dots,n\}$. By Lemma \ref{lemma_comp_to_word}, 
	there exists $w\in L(\cm,n)$ such that the prefix of $w$ is of the form 
	$$(s_I,\delta)(\#,\delta+\delta_1)\dots(\#,\delta+\delta_n)(l_1,\delta+1)$$ 
	and $\max(w)=n$. 	
	This together with Lemma \ref{lemma_function} implies that the suffix of $w$ is of the form $$(s_F,2k+\delta)(\sigma_1,2k+\delta+\delta_1)\dots(\sigma_n,2k+\delta+\delta_n)(\star,2k+\delta+1)$$ for some $\sigma_1,\dots,\sigma_n\in M\cup\{\#\}$.
	Note that in both the prefix and the suffix of $w$ the time delay between every symbol is $\delta_1$. 
	Define $\pv(p)=\delta_1$.
	It is easy to see that $w\in L_\pv(\A_\cm)$.
	Hence $L(\cm,n)\cap L_\pv(\A_\cm)\neq \emptyset$.

	For the direction from right to left, 
	assume there exist $n\in\N$ and a parameter valuation $\pv$ such that $L(\cm,n)\cap L_\pv(\A_\cm)\neq \emptyset$.
	Let $w\in L(\cm,n)\cap L_\pv(\A_\cm)$.
	By definition of $L(\cm, n)$, 
	the prefix of $w$ is of the form
	$$(s_I,\delta)(\#,\delta+\delta_1)(\#,\delta+\delta_2)\dots (\#,\delta+\delta_n)(l,\delta+1)$$
	for some $\delta\in\RP$, $0<\delta_1<\delta_2<\dots<\delta_n<1$, and $l\in L$.
	The clock constraints at the loop in location $2$ and at the edge from location $2$ to $3$ implies $\delta_i = {i\over{(n+1)}}$ for every $i\in\{1,\dots,n\}$ and $\pv(p)=\delta_1$.
	By Lemma \ref{lemma_function},
	the suffix of $w$ must be of the form  $$(s_F,N+\delta)(\sigma_n,N+\delta+\delta'_1)\dots(\sigma_m,N+\delta+\delta'_m)(\star,N+\delta+1)$$
	for some $N\in\N$,  $0<\delta'_1<\delta'_2<\dots<\delta'_m<1$ such that
	$n\leq m$, and there exists a strictly increasing function $f:\{1,\dots,n\}\to\{1,\dots,m\}$ such that $\delta_i = \delta'_{f(i)}$. 
	Note that $\star$ occurs exactly one time unit after $s_F$. 
	This, together with the clock constraints  at the loop in location $4$ and at the edge from $4$ to the final location $5$, 
	implies $m=n$ (and $\delta'_i=\delta_i$ for every $i\in\{1,\dots,n\}$). 
	By Lemma \ref{lemma_function}, we further know that the number of symbols between a control state symbol and a symbol in $L\cup\{\star\}$ cannot \emph{decrease}, and hence it follows that $\max(w)=n$. 
	By Lemma \ref{lemma_word_to_comp}, 
	there exists an error-free computation of $\cm$ from $(s_0,\varepsilon)$ to $(s_F,x)$ for some $x\in M^*$. 	
\end{proof}
\subsection{The Reduction}
We define $L(\cm)=\cup_{n\in\N} L(\cm,n)$. 
Then we obtain
\begin{corollary}
	\label{col}
	There exists  an error-free computation of $\cm$ from $(s_I,\varepsilon)$ to $(s_F,x)$ for some $x\in M^*$, if, and only if, there exists some parameter valuation $\pv$ with $L_\pv(\A_\cm)\cap L(\cm) \neq \emptyset$. 
\end{corollary}
Next, we define the MTL formula $\varphi_\cm$ such that $L(\varphi_\cm)=L(\cm)$. 
The formula $\varphi_\cm$ is the conjunction of a set of formulas, each of them expressing one of the conditions of $L(\cm)$.  
We start by  defining some auxiliary formulas:
$\bigvee S \defeq \bigvee_{s\in S} s$,
 	$\bigvee M \defeq \bigvee_{m\in M} m$,
	$\bigvee L \defeq \bigvee_{l\in L} l$,
	$\varphi_{\mathsf{copyM}} \defeq \glob_{(0,1)}\bigwedge_{m\in M} (m\rightarrow \finally_{=2} m)$, and
	$\varphi_{\mathsf{copy}\#} \defeq\glob_{(0,1)}( \# \rightarrow \finally_{=2} \#)$.
\begin{itemize}
\item $\glob(\X_{>0}\true \vee \neg\X\true)$ (Strict monotonicity) 
	
\item $\glob\langle\bigwedge_{s\in S\backslash\{s_F\}} (s\rightarrow\bigvee_{(s,l,s)\in\Delta} (\finally_{=1}l \wedge \finally_{=2} s')) \wedge (s_F\rightarrow \finally_{=1}\star)\rangle$, \newline $\glob\langle \bigvee S \rightarrow ((\glob_{<2}\neg\bigvee S) \wedge (\glob_{(0,1)\cup(1,2)}\neg\bigvee L)\rangle$ (Conditions on the occurrence of control state symbols and symbols in $L\cup\{\star\}$)
	
\item $\glob\langle
	\bigvee S \rightarrow (\glob_{(0,1)} (\bigvee M \vee \#) \wedge \glob_{[1,2)} \neg(\bigvee M \vee \#))\rangle$, $\glob((\#\wedge\X \bigvee M)\rightarrow \false)$ (Conditions on symbols in $M\cup\{\#\}$)

\item $s_I\wedge \bigvee_{(s_I,l,s)\in\Delta} (\#\U  (l \wedge \X s))$ (Encoding of the initial configuration)
	
\item $\finally s_F$ (Reaching $s_F$)
	
\item $\glob\bigwedge_{(s,\varepsilon,-)\in\Delta\atop s\neq s_F} ((s\wedge \finally_{=1} \varepsilon)\rightarrow ((\glob_{(0,1)}\neg\bigvee M )\wedge \varphi_{\mathsf{copy}\#}))$
	
\item $\glob \bigwedge_{\delta=(s,m!,-)\in\Delta\atop s\neq s_F} ( (s\wedge\finally_{=1}m!)\rightarrow (\varphi_{\mathsf{copyM}} \wedge \varphi_{next\#}\wedge\varphi_{yes\#}  \wedge \varphi_{no\#} ))$, where
		\begin{itemize}
		\item $\varphi_{next\#} = \X\# \rightarrow (\X\finally_{=2} m \wedge \X\varphi_{\mathsf{copy}\#})$
		\item $\varphi_{yes\#} = (\finally_{<1}\wedge\neg\X\#) \# \rightarrow \glob_{<1}((\neg\#\wedge \X\# )\rightarrow \X\finally_{=2} m \wedge \X\varphi_{\mathsf{copy}\#}))$
		\item $\varphi_{no\#} = \neg\finally_{<1}\# \rightarrow  \glob_{<1} (\X m! \rightarrow \finally_{=2}(\X m\wedge\X\X \bigvee L))$
		\end{itemize}

	\item $\glob \bigwedge_{(s,m?,-)\in\Delta\atop s\neq s_F} ((s \wedge\finally_{=1} m?) \rightarrow
	(\varphi_{\mathsf{yesm}} \wedge \varphi_{\mathsf{nom}}))$, where 
		\begin{itemize}
		\item $\varphi_{\mathsf{yesm}} =  \X m \rightarrow (\varphi_{\mathsf{shift}} \U m?)$, 
		 $\varphi_{\mathsf{shift}}= \bigwedge_{m\in M}( \X m \rightarrow \finally_{=2}m) \wedge ( \X\#\rightarrow \finally_{=2}\#) \wedge (\X m? \rightarrow \finally_{=2}\#)$
	 \item $\varphi_{\mathsf{nom}} = \X\neg m\rightarrow (\varphi_{\mathsf{copyM}} \wedge \varphi_{\mathsf{copy}\#} \wedge \glob_{<1}(\X m? \rightarrow \finally_{=2}(\X\#\wedge \X\X\bigvee L)))$
		\end{itemize}
\end{itemize}

	\paragraph{\bf Proof of Theorem \ref{thm_main}}
	Let $\cm=(S,s_0,M,\Delta)$ be a channel machine, let $s_F\in S$.
	Define the parametric timed automaton $\A_\cm$ and the MTL formula $\varphi_\cm$ as above. 
	By Corollary \ref{col} we know that there is an error-free computation from $(s_I,\varepsilon)$ to $(s_F,x)$ for some $x\in M^*$, if, and only if, there exists some parameter valuation $\pv$ with  $L_\pv(\A_\cm)\cap L(\varphi_\cm) \neq \emptyset$. 
	The latter, however, is equivalent to $L_\pv(\A_\cm)\not\subseteq L(\neg\varphi_\cm)$, \ie, there exists some timed word $w\in L_\pv(\A_\cm)$ such that $w\not\models\neg\varphi_\cm$. 
	Hence, the MTL-model checking problem for parametric timed automata is undecidable. \qed

	\section{Discussion}
	For our undecidability result we construct a parametric timed automaton using a parametric \emph{equality} constraint of the form $x=p$. 
	Parametric equality constraints	seem to be a source of undecidability; they occur in the undecidability proofs of, \eg, the emptiness problem for parametric timed automata with three clocks~\cite{DBLP:conf/stoc/AlurHV93}, and the satisfiability problem for a parametric extension of LTL~\cite{DBLP:journals/tocl/AlurETP01}. A natural question is thus to consider the MTL-model checking problem for L/U-automata~\cite{DBLP:journals/jlp/HuneRSV02}, a subclass of parametric timed automata in which parameters are only allowed to occur either as a lower bound or as an upper bound, but not both, and for which the emptiness problem is decidable independent of the number of clocks.
	We further remark that the proof does not work if we restrict the parameter valuation to be a function mapping each parameter to a non-negative integer.

	\paragraph{\emph{Acknowledgements}} I would like to thank James Worrell for pointing me to MTL's capability to encode computations of \emph{Turing machines with insertion errors}, explained in~\cite{DBLP:conf/fossacs/OuaknineW06}. 

\nocite{*}
\bibliographystyle{eptcs}

\end{document}